\documentclass[12 pt]{article}
\usepackage{bm}%

\usepackage{amssymb}

\usepackage{graphicx}
\usepackage{amsmath}
\usepackage{amsthm}
\usepackage{color}
\newcommand{\rf}[1]{(\ref{#1})}
\newcommand{\dst}{\displaystyle}
\newcommand{\R}{\mathbb {R}}
\newcommand{\E}{\mbox{\sf{E}}}
\newcommand{\M}{\bm M}
\newcommand{\Mt}{\tilde{\M}}
\renewcommand{\L}{\mbox{$\cal L$}}

\newcommand{\A}{\bm A}
\newcommand{\B}{\bm B}
\newcommand{\D}{\bm D}
\newcommand{\I}{\bm I}

\newcommand{\T}{\bm T}
\renewcommand{\v}{\bm v}
\newcommand{\V}{\bm V}
\renewcommand{\u}{\bm u}

\newcommand{\nonu}{\nonumber}
\newcommand{\Tr}{\mbox{\sf tr\,}}
\newcommand{\e}{\eta}
\newcommand{\xd}{\xi_{\Delta}}
\renewcommand{\det}{\mbox{\sf det\,}}

\renewcommand{\leq}{\leqslant}
\renewcommand{\geq}{\geqslant}
\renewcommand{\tilde}{\widetilde}
\newcommand{\hs}{\textsc {hs}}

\newtheorem{proposition}{Proposition}
\newtheorem{lemma}{Lemma}

\begin{document}

\author{Yuri A. Godin, Stanislav Molchanov, Boris Vainberg \\
Department of Mathematics and Statistics \\
University of North Carolina at Charlotte \\
Charlotte, NC 28223, USA \\
}

\title{The effect of disorder on the wave propagation \\
in one-dimensional periodic optical systems}

\maketitle

%\date{\today}

\begin{abstract}
% The influence of disorder on the bandgap structure and wave propagation in periodic (optical) waveguides is studied. We show that random perturbation 
% of the transfer matrix $\M$ can be represented in a universal form $\Mt = \M (\I + \sigma \V + O(\sigma^2))$, where $\sigma$ characterizes strength of disorder
% and elements of random matrix $\V$ have Gaussian distribution with non-degenerate
% covariance matrix and $\Tr \V =0$. Lyapunov exponents of such matrices estimate each other that allows to evaluate the Lyapunov exponent of random waveguide through that of the Schr\"{o}dinger equation with random potential.
% Dependence of the Lyapunov exponent $\gamma$ on the frequency $\omega$ and magnitude of the disorder $\sigma$ is also investigated. It is shown that in the bulk of the bands $\gamma \sim \sigma^2$, while near the band edges it has the order $\gamma  \sim \sigma^{2/3}$. This dependence is illustrated by numerical simulation.
The influence of disorder on the transmission through periodic waveguides is studied. Using a canonical form of the transfer matrix we investigate dependence of the Lyapunov exponent $\gamma$ on the frequency $\nu$ and magnitude of the disorder $\sigma$. It is shown that in the bulk of the bands $\gamma \sim \sigma^2$, while near the band edges it has the order $\gamma  \sim \sigma^{2/3}$. This dependence is illustrated by numerical simulations.
\end{abstract}

%\pacs{42.25.Dd, 43.20.Bi}% PACS, the Physics and Astronomy
                             % Classification Scheme.
%\keywords{disorder, periodic-on-average, Lyapunov exponent, transfer matrix, transmission coefficient, periodic waveguide}

%\begin{abstract}
% We study the variance of the solution of a periodic randomly perturbed 1D
% Schr\"{o}dinger operator after the propagation through $N$ periods.
%  It is shown that if the frequency of propagation lies inside the band, then the
% total variance is proportional to $N\sigma^2$, where $\sigma$ is the intensity
% of the white noise. However, if the wave frequency is close to the band edge
% (where the transfer matrix has a Jordan block structure), the resulting
% variance is proportional to $N\sigma^{2/3}$. Thus, propagation becomes highly
% sensitive to random perturbations.
%
% Numerical simulations reveal that even small noise in periodic potential can
% suppress transmission near the band edges and make it strongly irregular inside
% the band. Further increase of the noise amplitude
% leads to intermittent behaviour of the transmission coefficient, and
% makes transmission possible only for few random frequencies
% in the band.
%\end{abstract}

\maketitle
%\markboth{\hfill \today \hfill}{\hfill \today \hfill}

\section{Introduction}
\label{I}

Propagation of time-harmonic electromagnetic waves in one-dimensional periodic optical waveguides can often be reduced to the eigenvalue problem
\begin{equation}
 -\frac{1}{n^2(x)}\frac{d^2 \psi}{dx^2}=\nu^2 \psi(x),
  \label{evp}
\end{equation}
where $\psi (x)$ is a Cartesian component of electric or magnetic field, $n(x)=n(x+\ell)$ is the index of refraction which we suppose to be periodic.
%and$\nu = \omega/c$ is the ratio of frequency and speed of light.
A way to find
the frequencies $\nu$ for which waves can propagate is by computing the so-called
characteristic or transfer matrix $\bm M$ that maps two-dimensional vector with components $\psi$ and $\psi^\prime $ at the beginning of the period to the value of that  vector at the end of the period
\begin{equation}\label{tran}
 {\bm M}\left[
\begin{array}{c}
 \psi(0) \\[2mm]
\psi^\prime (0)
\end{array}
\right] =
\left[
\begin{array}{c}
 \psi(\ell) \\[2mm]
 \psi^\prime (\ell)
\end{array}
\right].
\end{equation}
%We use the transfer matrix in the Pr\"{u}ffer form, since many formulas below involving matrix $\M$ are valid only for %the Pr\"{u}ffer form of $\M$. The regular transfer matrix does not include parameter $\nu$ in (\ref{tran}).
%Near the band edges, however, and, in particular, for $\nu \approx 0$ we will use the standard transfer matrix relating $\psi$ and $\psi^\prime$ over the period, see sec. 4).

Matrix ${\bm M}$ is an analytic function of the frequency $\nu$ and $\det M = 1$. Propagation frequencies $\nu$ of the problem \rf{evp} satisfy the condition
\begin{equation}
 \left| \Tr {\bm M} \right| < 2.
 \label{tr}
\end{equation}
If \rf{tr} does not hold for some frequency $\nu$ then the corresponding wave cannot propagate. The spectrum of operator \rf{evp} (which is the closure of the set $\{\nu^2\}$ for all propagating frequencies) has a bandgap structure with the bands defined by $\left| \Tr {\bm M} \right| \leq 2$.

In reality, however, periodic waveguides contain finite number of periods $N$. Equation \rf{evp} in this case becomes
\begin{equation}
 -\frac{1}{n_{N}^2(x)}\frac{d^2 \psi}{dx^2}=\nu^2 \psi(x),
  \label{evpn}
\end{equation}
where
\begin{equation}
 n_{N}(x) = \left\{
 \begin{array}{cl}
 1, & x \notin [0, N\ell], \\[2mm]
 n(x), & x \in [0,N\ell].
 \end{array}
 \right.
\end{equation}
Such waveguide supports any frequency $\nu \geq 0$, but transmission for different
frequencies varies. The transmission and reflection coefficients $t_N(\nu)$ and $r_N(\nu)$ are determined by solving the scattering problem
\begin{equation}
 \frac{d^2 \psi}{dx^2}+\nu^2n_{N}^2(x) \psi(x)=0
\end{equation}
with
\begin{equation}
\psi (x)=\left\{
\begin{array}{ll}
e^{i\nu x}+r_N\, e^{-i\nu x}, & x<0, \\
t_N \,e^{i\nu x}, & x>N\ell.
\end{array}
\right.
\label{t}
\end{equation}
The transfer matrix  over the interval $(0,N\ell)$ is ${\bm M}(N\ell)= {\bm M}^N $. The transfer matrix $\M_P$ in the Pr\"{u}ffer form maps two-dimensional vector with components $\psi$ and $\psi^\prime/\nu $ at the beginning of the interval to the value of that vector at the end of the interval, i.e.
\[
{\bm M_P}(N\ell)= \T^{-1}{\bm M}^N \T,~~\T =\left[
\begin{array}{cc}
|\nu|^{-1/2}& 0 \\
0 & |\nu|^{1/2}
\end{array}\right].
\]
This matrix defines the transmission coefficient \cite{LGP88,MV07}
\begin{equation}
|t_N |^{2}=\frac{4}{\left\| {\bm M_P^N}(N\ell)\right\| _{\hs}^{2}+2}=\frac{4}{\left\| \T^{-1} {\bm M}^N \T\right\| _{\hs}^{2}+2},
\label{tn}
\end{equation}
where $\|{\bm M}\|_{\hs}$ is the Hilbert-Schmidt norm of the matrix ${\bm M}=[M_{i,j}]$
\begin{equation*}
\left\| {\bm M}\right\| _{\hs}^{2}=\Tr{\bm M}{\bm M}^{\ast}=\sum_{i,j=1}^{2} M_{i,j}^{2}.
\end{equation*}

Random imperfections in periodic waveguides can significantly affect their
properties, especially near degenerate band edge, see \cite{FV06,FV06a}.
This paper deals with the influence of disorder
on the propagation properties of the waveguide. Let $\Mt_k$, $k=1,2,\ldots,N$ be
transfer matrices of the disordered waveguide on the periods $[(k-1)\ell, k\ell]$. Then formula \rf{tn} becomes
\begin{equation}
 |t_N |^{2}=\frac{4}{\left\| \T^{-1}\Mt_N \Mt_{N-1} \ldots \Mt_1 \T\right\| _{\hs}^{2}+2}.
\label{tnt}
\end{equation}
The norm of the product of random matrices is related to the Lyapunov exponent
\begin{equation}
 \gamma =  \lim_{N \to \infty} \frac{\ln\| \Mt_N \Mt_{N-1} \ldots \Mt_1\| _{\hs}}{N}.
\label{le}
\end{equation}
If this limit exists then
\begin{equation}
| t_N | \sim e^{-\gamma N}, \quad N \to \infty.
 \label{t-gamma}
\end{equation}

Thus, the Lyapunov exponent allows one to estimate the energy transmitted through $N$ periods of the waveguide. In the paper we study $\gamma$ as a function of disorder strength $\sigma$ and show that it has different behavior depending on the location of frequency $\nu$ in the band.

We will consider two types of disorder. Disorder of the first type concerns perturbation of physical properties of
a periodic waveguide and is described by the equation
\begin{equation}
 -\frac{d^2 \psi}{dx^2}+ \left[ \nu^2 A(x) + B(x)\right] \psi + \sigma \frac{\xd (x)}{\sqrt\Delta} \psi = 0, ~~A(x+\ell)=A(x),~B(x+\ell)=B(x),
 %\label{eq2}
\end{equation}
where $\xd$ is a stationary random noise with correlation length $\Delta \ll 1,$ $ \E \,\xd =0$,
$ |\xd|<1$, and $\sigma\ll\sqrt\Delta,$ i.e.
$\frac{\xd (x)}{\sqrt\Delta}$ converges in law to the white noise as $\Delta \to 0$. Disorder of the second type is related to perturbation of the length of the period.

We shall show that in both cases matrices $\Mt_k$ can be represented (c.f. \cite{GMV07}) as a
multiplicative perturbation of $\M$
\begin{equation}
 \tilde{\M}_k = {\M}\left( \I + \sigma \,\V_k + O(\sigma^2)\right),
 \label{Mt1}
\end{equation}
where $\Tr \V_k=0$:
\begin{equation}\label{vk}
\V_k = \left[
\begin{array}{rr}
\xi_k & \e_k \\[2mm]
\zeta_k & -\xi_k
\end{array}
\right], \quad k \geq 1,
\end{equation}
and entries $\xi_k, \e_k,\zeta_k$ of $\V_k$ in the limit as $\Delta \to 0$ form Gaussian random vectors independent for different $k$.
The covariance matrix
\begin{equation}
{\B_\Delta} = \E \left[
\begin{array}{ccc}
\xi_k^2 & \xi_k \eta_k & \xi_k \zeta_k \\[2mm]
\xi_k \eta_k & \eta_k^2 & \eta_k \zeta_k \\[2mm]
\xi_k \zeta_k  & \eta_k \zeta_k & \zeta_k^2
\end{array}
\right]
\label{B}
\end{equation}
of elements $\xi_k , \e_k ,
\zeta_k$ is non-degenerate for small $\Delta>0$; moreover, the limit $\B_0=\lim_{\Delta \to 0}\B_\Delta$ exists and
$\det{\B_0} > 0$.
Later we will refer mostly to \rf{Mt1} rather than a specific form of the equation. Thus our results can be applied to more general periodic media with small and short correlated disorder.

%transfer matrix $\M (N\ell)=\M^N(\ell)$.
%Exponential decay of $t_N$ as $N \to \infty$ is related with
% We derive asymptotics of the Lyapunov exponent $\gamma$ for small strength of disorder $\sigma \ll 1$ and show that it  depends
% essentially on the location of frequency $\nu$ in the optical band. Then formula \rf{tn} allows us to establish the dependence of the transmission coefficient $t_N$ on $\sigma$.
% Here we have in mind two specific models depending on the parameters of the waveguide which are allowed to be random.

The paper contains the following results. Section 2 concerns the deterministic problem. We discuss bandgap structure of the spectrum of the periodic problem \rf{evp} and properties of the transfer matrix $\M$ near the fixed band edge. We provide an example showing that unlike the Schr\"{o}dinger equation, the gaps of the optical equation \rf{evp} are located irregularly and their length does not go to zero. After that throughout the paper we will consider $\nu$ only in a fixed band. Then we show that the transfer matrix in a neighborhood of the band edge can be reduced to a simple canonical form defined by equation \rf{evp} with $n=1$.
In Section 3 we show that the canonical representation of the transfer
matrix holds for two major cases: the disorder of the refraction index and length of the period.
 Then in Section 4 we estimate the Lyapunov exponent of randomly perturbed periodic waveguide by
 the Lyapunov exponent of the simplest model -- the Schr\"{o}dinger equation with white noise
 potential. Section 5 contains asymptotic analysis of the white noise model for the Schr\"{o}dinger
 equation and derivation of expression for the Lyapunov exponent. We show that $\gamma\sim \sigma^2$
 in the bulk of the band and $\gamma\sim \sigma^{2/3}$ near non-degenerate band edges,
 and present results of numerical simulations. This asymptotics for the Schr\"{o}dinger equation leads to a similar estimate for the optical problem due to results of Section 4.
%Finally, in Section 5 we discuss two particular models of disordered waveguides and present results of numerical simulations.

\section{Spectral properties of periodic optical waveguides}
\label{II}

We start with an example showing the main difference between the optical and Schr\"{o}dinger equations.
%\setcounter{equation}{0}
% The bandgap structure of periodic optical waveguide described by equation \rf{evp} differs essentially from that of the Schr\"{o}dinger equation. To illustrate this special feature we consider two examples.
Let the index of refraction $n(x)$ be $\ell$-periodic, piecewise constant with two different values $n_1, n_2$ on two successive subintervals of length $l_1, l_2$ on the interval of periodicity $ \ell=l_1+l_2$.
Then
\begin{equation}
 \M = \M_2 \M_1,
\end{equation}
where
\begin{equation}
 \M_i = \left[
 \begin{array}{cc}
  \cos ( n_{i} l_i \nu) & \dst \frac{\sin ( n_{i} l_i \nu)}{n_{i}\nu} \\[2mm]
  -n_{i}\nu\sin ( n_{i} l_i \nu) & \cos ( n_{i} l_i \nu)
 \end{array}
 \right], \;i=1,2.
\end{equation}
Hence
\begin{align}
 \Tr \M  &= 2\cos ( n_{1} l_1 \nu) \cos ( n_{2} l_2 \nu) \nonu\\[2mm]
&-\left(\frac{n_1}{n_2} + \frac{n_2}{n_1} \right)\sin ( n_{1} l_1 \nu ) \sin ( n_{2} l_2 \nu) \nonu \\[2mm]
&=A\cos \left( a\nu \right) - B\cos \left(b\nu\right),
\end{align}
with $\dst A = \frac{1}{2}\left(\frac{n_1}{n_2} + \frac{n_2}{n_1} \right)+1, \;
B = \frac{1}{2}\left(\frac{n_1}{n_2} + \frac{n_2}{n_1} \right)-1,\; a = n_{1}l_{1} + n_{2}l_{2}, b = n_{1}l_{1} - n_{2}l_{2}.$ If numbers $n_{1}l_{1}$
and $n_{2}l_{2}$ are rationally independent then almost periodic function $f(\nu)=\Tr \M $ oscillates along the $\nu$ axis between $\dst A+B = \frac{n_1}{n_2} + \frac{n_2}{n_1} > 2$ and $-(A+B) < -2$. In this case there are infinitely many
gaps in the spectrum and their length does not vanish as $\nu \to \infty$. In fact, there are values
 $\nu=\nu_k \rightarrow \infty$ as $i\rightarrow \infty,$ such that $\cos a\nu_k=1, ~\cos b\nu_k \approx -1,~f(\nu_k)>\frac{1}{2}(A+B+2)>2$ and $f_{\nu}$ is bounded uniformly in $\nu$.
At the same time the gaps can be
arbitrary small, see Fig. \ref{disc_op}.

\begin{figure}[ht]
\begin{center}
\includegraphics[width=0.8\textwidth, angle=0]{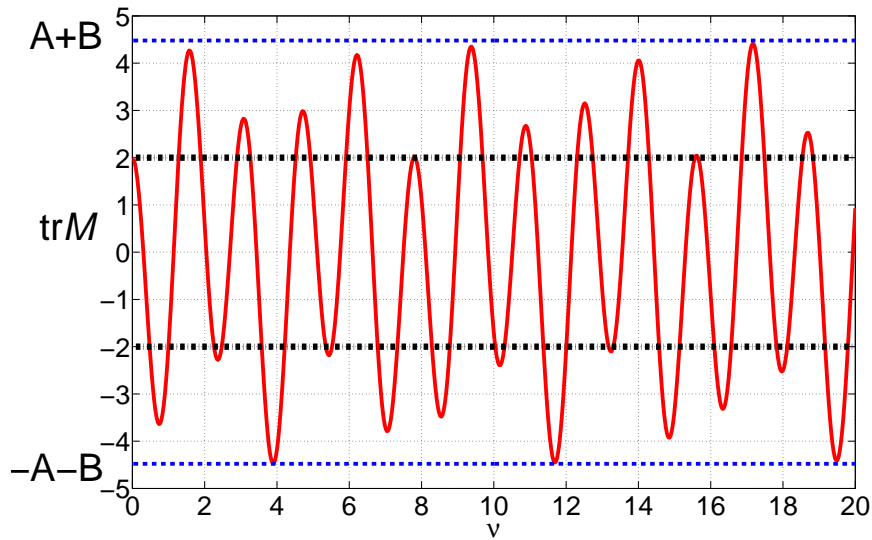}
\end{center}
\caption{Graph of the discriminant curve of the periodic optical waveguide with $l_1=2,\;l_2=0.2,\;n_1=\sqrt{2},\;n_2=6$ oscillates between $\dst A+B = \frac{n_1}{n_2} + \frac{n_2}{n_1}$ and $-(A+B)$ exhibiting infinite number of bands and gaps as $\nu \to \infty$.}
\label{disc_op}
\end{figure}

We remind that in the case of periodic Schr\"{o}dinger equation
\begin{equation}
 -\frac{d^2 \psi}{dx^2} + V(x)\psi(x)=\nu^2 \psi(x), \quad V(x)=V(x+\ell),
  \label{schr}
\end{equation}
the gaps are located periodically (see Fig. \ref{disc_sh}) and asymptotically vanishing \cite{M86}. For the
$n$-th gap $g_n$ one has
\begin{equation}
 g_n = \left( \frac{\pi n}{\ell} + \epsilon_n, \frac{\pi n}{\ell} + \delta_n, \right), \quad \epsilon_n, \delta_n \to 0 ~~\text{ as} ~~ n \to \infty.
\end{equation}
Asymptotic behavior of $\epsilon_n,\delta_n,$ depends on the smoothness of potential $V(x)$ (see \cite{M86}).

\begin{figure}[ht]
\begin{center}
\includegraphics[width=0.8\textwidth, angle=0]{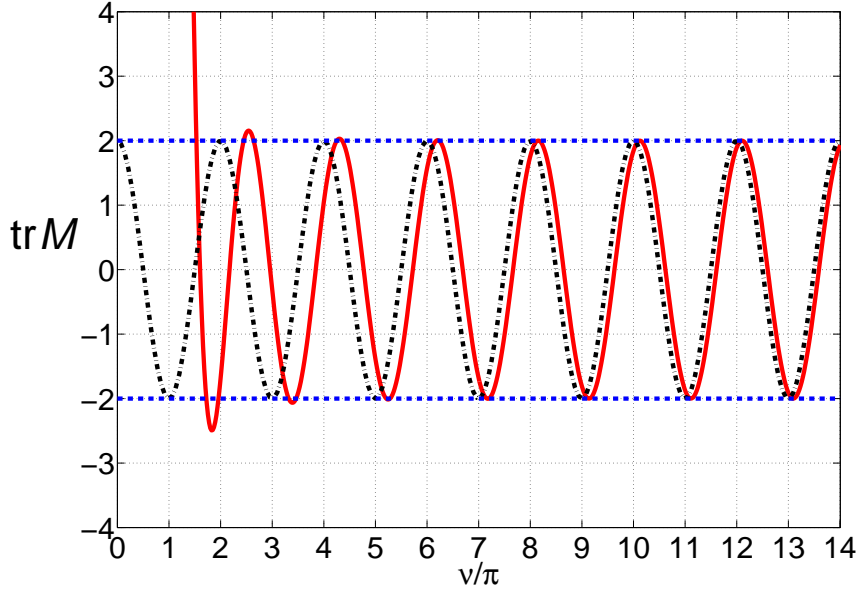}
\end{center}
\caption{Graph of the discriminant curve of the periodic Schr\"{o}dinger equation (solid line) with $l_1=2,\;l_2=0.2,\;n_1=4,\;n_2=9$. Dashed line corresponds to $2\cos \nu \ell$.}
%oscillates between $\dst A+B = \frac{n_1}{n_2} + \frac{n_2}{n_1}$ and $-(A+B)$ exhibiting infinite number of bands and gaps as $\nu \to \infty$.}
\label{disc_sh}
\end{figure}

Our next result concerns the structure of the transfer matrix near the band edges. Let us consider one of the spectral
bands $\beta = [\nu_0, \nu_1]$ and analyze the transfer matrix $\M$ as a function of $\nu$ when $\nu \to \nu_0$. The case
$\nu \to \nu_1$ can be considered similarly.
Inside the band $\beta$ we have
\begin{equation}
 \left| \frac{1}{2}\, \Tr \M (\nu) \right | < 1.
\end{equation}
The eigenvalues of $\M$ have the form $e^{\pm i \omega},~\omega = \omega (\nu),$ and hence
\begin{equation}\label{cos}
\cos \omega = \frac{1}{2}\, \Tr \M (\nu).
\end{equation}
This matrix can be reduced to the diagonal form for any $\nu$ strictly inside the band.
%Consider the equation $\psi^{\prime \prime} + \nu^2 \psi = 0$, $x \in [0,1]$, $\nu \ll 1$. The transfer matrix
%$\M(\nu)$ in the Pr\"{u}ffer form
%\begin{equation}
% \M(\nu) = \left[
% \begin{array}{cc}
% \cos \nu & \sin \nu \\[2mm]
% -\sin \nu & \cos \nu
% \end{array}
% \right]
%\end{equation}
%becomes the identity matrix as $\nu \to 0$ while regular transfer matrix
%\begin{equation}
% \M(\nu) = \left[
% \begin{array}{cc}
% \cos \nu & \frac{\sin \nu}{\nu} \\[2mm]
% -\nu \sin \nu & \cos \nu
% \end{array}
% \right]
%\end{equation}
%turns into Jordan block
%$\left[
% \begin{array}{cc}
% 1 & 1 \\
% 0 & 1
% \end{array}
 %\right]
%$
%when $\nu \to 0$.
At the band edges, however, both eigenvalues of $\M$ equal either $+1$ or $-1$. Throughout
the paper we consider only non-degenerated band edges, i.e., we assume that the derivative of the function
$\Tr \M (\nu)$ is not zero at the band edge (where $|\Tr \M (\nu)|=2$). It means that we do not consider points
like $\nu=7.755$ in Fig. 1. Equation \rf{cos} immediately implies that $\omega$ is an analytic function
of $\sqrt{\nu-\nu_0}$  in a neighborhood of a non-degenerate band edge $\nu=\nu_0$. It is also important that
the transfer matrix cannot be diagonalized at a non-degenerate band edge, but rather reduced to the Jordan form.

The simplest example of the discussed situation is given by
the equation $\psi^{\prime \prime} + \omega^2 \psi = 0$, {\bf $0 \leq x \leq 1$ }. The transfer matrix for that equation is
\begin{equation}
 \M(\omega) = \left[
 \begin{array}{cc}
 \cos \omega & \frac{\sin \omega}{\omega} \\[2mm]
 -\omega \sin \omega & \cos \omega
 \end{array}
 \right],
\end{equation}
and it turns into the Jordan block at the non-degenerate band edge $\omega=0$. The next Proposition
shows that arbitrary transfer matrix in a neighborhood of a non-degenerate band edge can be reduced to
that simple form uniformly in $\omega$.
%points (take for example $\nu_0$) matrix $\M(\nu_0)$ has Jordan block
%\begin{equation}
% \M(\nu_0) = \D^{-1}_{0} \left[
% \begin{array}{cc}
% 1 & 1\\
% 0 & 1
% \end{array}
% \right] \D_{0}
%\end{equation}
%for an appropriate matrix $\D_{0}$, $\det \D_{0} = 1$.
%\newpage
\begin{proposition}
\label{prop1}
Let $\M=\M(\nu)$ be a transfer matrix of equation \rf{evp} in a neighborhood of a non-degenerate band edge
$\nu=\nu_0$, and let $e^{\pm i\omega}$ be
the eigenvalues of  $\M(\nu),~\omega=\omega (\nu),~\omega (\nu_0)=0.$
Then there exists an analytic matrix-valued function $\D (\omega),~\det \D \neq 0$, such that
\begin{equation}
 \widehat{\M}=\D^{-1} \M(\nu) \D = \left[
 \begin{array}{cc}
 \cos \omega & \frac{\sin \omega}{\omega}\\
 -\omega \sin \omega & \cos \omega
 \end{array}
 \right].
 \label{p1}
\end{equation}
\end{proposition}
\noindent
\textit{Remark.} This proposition is valid for general matrices with the same properties as those mentioned
above for $\M(\nu)$.
\begin{proof}
 Let
\begin{equation}
 \M(\nu) = \left[
 \begin{array}{cc}
 a & b\\
 c & d
 \end{array}
 \right], \quad \frac{1}{2}\, \Tr \M (\nu)= \frac{a+d}{2} = \cos \omega,
\end{equation}
where $\omega=\omega(\nu)$ is real for $\nu >\nu_0$ (in the band) and $\omega=\omega(\nu)$ is complex
for $\nu <\nu_0$ (in the gap). We have that $|b|+|c| > 0$ for $\nu =\nu_0$, since $\M (\nu)$ is not diagonalizable.
We may assume that $b(\nu_0) \neq 0$, and therefore
$|b(\nu)| \geq \delta > 0$ for $|\nu-\nu_0| \leq \epsilon$.

The eigenvectors of matrix $\M$ corresponding to the eigenvalues $ e^{\pm i\omega}$ have the form
\begin{equation}
 \v_{1,2}(\nu) = \left[
 \begin{array}{c}
 -b\\
 a - e^{\pm i\omega}
 \end{array}
 \right].
\end{equation}
Vectors $\v_{1,2}(\nu)$ form a basis for any $\omega > 0$. However, $\v_1 = \v_2$ when $\omega = 0,$ and
the basis degenerates.

In order to avoid the degeneracy, we introduce new basis
\begin{equation}
 \u_1 = \frac{1}{2} (\v_1 + \v_2), \quad
 \u_2 = \frac{1}{2i\omega}(\v_1 - \v_2).
\end{equation}
The new basis is uniformly non-degenerate, and matrix $\M$ in this basis has the form
\begin{equation}
 \left[
 \begin{array}{cc}
 \cos \omega & \frac{\sin \omega}{\omega}\\
 -\omega \sin \omega & \cos \omega
 \end{array}
 \right],
\end{equation}
i.e. \rf{p1} holds with
\begin{equation}
 \D= \left[
 \begin{array}{cc}
 -b & 0\\
 a -\cos \omega & -\frac{\sin \omega}{\omega}
 \end{array}
 \right].
\end{equation}
\end{proof}

\section{Transfer matrix for disordered wave guide}
\label{III}

Here we will show that our main assumption \rf{Mt1} on the form of the transfer matrix holds for
typical random perturbations of
deterministic equations. Two classes of problems will be considered. The first one concerns the
 equation
\begin{equation}
 -\frac{d^2 \psi}{dx^2}+ \left[ \nu^2 A(x) + B(x)\right] \psi = 0, ~~A(x+\ell)=A(x),~B(x+\ell)=B(x),
 \label{eq1}
\end{equation}
%on the interval $[0,\ell]$ for the frequency $\nu = \nu_0 >0$ and
and its random perturbation
\begin{equation}
 -\frac{d^2 \psi}{dx^2}+ \left[ \nu^2 A(x) + B(x)\right] \psi + \sigma  \frac{\xd (x)}{\sqrt\Delta} \psi = 0.
 \label{eq2}
\end{equation}
%{\color{red}{(for example, refraction index disorder)}}
Here $\xd$ is a stationary random process with short
correlation length $\Delta \ll \ell$, $\E \,\xd =0$, which converges to the white noise as $\Delta \to 0$, and
$\sigma \ll \sqrt \Delta$ is a small parameter (strength of disorder).

Let $\psi_1 (x),~\psi_2 (x)$ be the fundamental solutions of equation \rf{eq1} with $\psi_1 (0)=1,\;
\psi_{1}^\prime (0) = 0, \;\psi_2 (0) = 0,\; \psi_2^\prime (0) = 1$,
and  let $\widetilde{\psi}_1 (x),~\widetilde{\psi}_2 (x)$ be similar fundamental solutions of the perturbed
equation \rf{eq2}.
The transfer matrices on the interval $[0,\ell]$ for \rf{eq1}, \rf{eq2} are given by
\begin{equation}
 \M (\nu) =\left[
 \begin{array}{cc}
 \psi_1 (\ell) & \psi_2 (\ell)\\[2mm]
 \psi_1^\prime (\ell) & \psi_2^\prime (\ell)
 \end{array}
 \right],~~~\widetilde{\M} (\nu) =\left[
 \begin{array}{cc}
 \widetilde{\psi}_1 (\ell) & \widetilde{\psi}_2 (\ell)\\[2mm]
 \widetilde{\psi}_1^\prime (\ell) & \widetilde{\psi}_2^\prime (\ell)
 \end{array}
 \right].
\end{equation}
Transfer matrix $\Mt_k (\nu)$ on the interval $[ k\ell,(k+1)\ell ]$ is defined similarly (matrix $\M (\nu)$ does not depend on $k$).
\begin{lemma}
\label{lem1}
 The perturbed transfer matrix $\Mt_k (\nu)$ of equation $\rf{eq2}$ on the interval $[k\ell ,(k+1)\ell ]$ has the form
 \begin{equation}
 \tilde{\M}_k = {\M}\left( \I + \sigma \V_k + O(\sigma^2)\right).
\label{Mt}
\end{equation}
Here matrix $\V_k$ has the form \rf{vk}, $\Tr\V_k=0$, and the covariance matrix $\B_\Delta$ \rf{B} has a
limit as $\Delta \to 0$,
$\B_\Delta = \B + O(\Delta)$, where $\B$ (and therefore $\B_\Delta$ for small $\Delta$) is non-degenerate.
\end{lemma}
\begin{proof} Without loss of generality we assume that $k=0$ and omit subindex $k$ in the proof.
 We compute $\Mt(\nu)$ using first order perturbation theory. Solutions $\tilde{\psi}_{1,2}$ for the disordered waveguide can be found from the integral equation
\begin{align}
\tilde{\psi}_i (x) &= \psi_i(x) + \sigma \int_0^x G(x,s) \frac{\xd (s)}{\sqrt\Delta} \tilde{\psi}_i (s) \,ds,
\label{psit}
\end{align}
where
\begin{equation}
G(x,s) = \left\{
\begin{array}{cl}
 \psi_1 (s) \psi_2 (x) - \psi_1 (x) \psi_2 (s),& s<x; \\[2mm]
 0, & s>x.
\end{array}
\right.
\end{equation}
Solution of \rf{psit} in the linear approximation has the form
\begin{align}
\tilde{\psi}_i (x) &= \psi_i(x) + \sigma \int_0^x G(x,s) \frac{\xd (s)}{\sqrt\Delta} \psi_i (s) \,ds.
\end{align}
If we put
\begin{equation}
 \eta_{ij}=\int_{0}^\ell \psi_i (s)\, \psi_j (s)\,\frac{\xd (s)}{\sqrt\Delta}\,ds, \quad i,j=1,2,
\end{equation}
then the transfer matrix of the disordered waveguide becomes
\begin{align}
 \tilde{\M} &= \left[
\begin{array}{cc}
 \tilde{\psi}_1 (\ell) & \tilde{\psi}_2 (\ell) \\[2mm]
 \tilde{\psi}_1^\prime (\ell) & \tilde{\psi}_2^\prime
\end{array}
\right] = \left[
\begin{array}{cc}
 \psi_1 (\ell) & \psi_2 (\ell) \\[2mm]
 \psi_1^\prime (\ell) & {\psi}_2^\prime
\end{array}
\right] \nonu \\[2mm]
&+\sigma \left[
\begin{array}{cc}
 \eta_{11}\psi_2 (\ell)-\eta_{12}\psi_1(\ell) & \eta_{12}\psi_2 (\ell) -\eta_{22} \psi_1 (\ell) \\[2mm]
 \eta_{11}\psi_2^\prime (\ell) -\eta_{12} \psi_1^\prime (\ell))  & \eta_{12} \psi_2^\prime (\ell) - \eta_{22}\psi_1^\prime (\ell)
\end{array}
\right]+ O(\sigma^2) \nonu \\[2mm]
& = \M + \sigma \M \left[
\begin{array}{cc}
 -\eta_{12} & -\eta_{22} \\[2mm]
 \eta_{11} & \eta_{12}
\end{array}
\right] + O(\sigma^2)
= \M \left(\I + \sigma \V + O(\sigma^2) \right),
\end{align}
where $\M$ is the unperturbed transfer matrix
and $\V = \left[
\begin{array}{cc}
 -\eta_{12} & -\eta_{22} \\[2mm]
 \eta_{11} & \eta_{12}
\end{array}
\right]$.
% \begin{align}
%  &\V = \M^{-1} \left[
% \begin{array}{cc}
%  \eta_{11}\psi_2 (\ell)-\eta_{12}\psi_1(\ell) & \eta_{12}\psi_2 (\ell)-\eta_{22} \psi_1 (\ell) \\[2mm]
%  \eta_{11}\psi_2^\prime (\ell) -\eta_{12} \psi_1^\prime (\ell) & \eta_{12} \psi_2^\prime (\ell) - \eta_{22}\psi_1^\prime (\ell)
% \end{array}
% \right]
% = \left[
% \begin{array}{cc}
%  -\eta_{12} & -\eta_{22} \\[2mm]
%  \eta_{11} & \eta_{12}
% \end{array}
% \right].
% \end{align}
%Here we used the fact that the Wronskian of the fundamental set of solutions $W(\psi_1 (x), \psi_2 (x)) = 1$.
Random variables $\e_{ij}$ are asymptotically Gaussian when $\Delta \to 0$ as it follow from the
central limit theorem. Their covariance matrix $\B_\Delta$ \rf{B} has a limit as $\Delta \to 0$,
$\B_\Delta = \B + O(\Delta)$, where $\B$ is expressed
in terms of the fundamental solution as follows
\begin{equation}
{\B} = \left[
\begin{array}{rrr}
\int_0^\ell \psi_1^2 (s) \psi_2^2 (s)\,ds & \int_0^\ell \psi_1 (s) \psi_2^3 (s)\,ds & -\int_0^\ell \psi_1^3 (s) \psi_2 (s)\,ds\\[2mm]
\int_0^\ell \psi_1 (s) \psi_2^3 (s)\,ds & \int_0^\ell \psi_2^4 (s)\,ds & -\int_0^\ell \psi_1 (s) \psi_2^3 (s)\,ds\\[2mm]
-\int_0^\ell \psi_1^3 (s) \psi_2 (s)\,ds & -\int_0^\ell \psi_1 (s) \psi_2^3 (s)\,ds&
\int_0^\ell \psi_1^4 (s)\,ds
\end{array}
\right].
\label{cov}
\end{equation}

Observe that $\B$ is the Gram matrix for the system of functions $-\psi^2_1 (s)$, $\psi_1 (s) \psi_2 (s)$,
$\psi^2_2 (s)$ on $[0,\ell]$.
The degeneracy of $\B$ would imply that this system of functions is
linearly dependent, i.e.,
\begin{equation}
 c_{11}\psi^2_1 (s) + c_{12} \psi_1 (s) \psi_2 (s) + c_{22} \psi^2_2 (s) =0,  \quad s \in [0,\ell],
\end{equation}
for appropriate constants $c_{11},c_{12},c_{22}$. The latter
is possible only if $\psi_1$ and $\psi_2$ are proportional, which
contradicts to linear independence of  $\psi_1$ and $\psi_2$.
Thus, the correlation matrix $\B$ is non-degenerate and so does $\B_\Delta$ for sufficiently small $\Delta$.
\end{proof}

Let us show that \rf{Mt1}-\rf{B} holds for another important class of problems related to the wave propagation in a
layered medium with layer thickness disorder.
Consider a periodic waveguide described by \rf{evp} with $\ell$-periodic piecewise-constant function $n(x)$.
Each period $[k \ell, (k+1)\ell]$ consists of $m$ subintervals of the length $l_i,~\ell=l_1+l_2+\ldots +l_m$, where $n(x)$
is equal to $n_i$ on $i$-th subinterval. Suppose that each subinterval changes randomly to
$l_i(1+\sigma \xi_i)$, where $|\xi_i| < 1$ are independent random variables with densities (for example, $\xi_i$
are uniformly distributed on the interval $(-1,1)$) and $\sigma$ is a small parameter. We assume that
distributions of $\xi_i$ do not depend on $k$ and omit index $k$ below.
Then the transfer matrix of $i$-th subinterval becomes

 \begin{equation}
  {\Mt}_i = \left[
 \begin{array}{cc}
  \cos n_i \nu l_i(1 + \sigma \xi_i) & \dfrac{\sin n_i \nu l_i(1 + \sigma \xi_i)}{n_i \nu} \\[2mm]
  -n_i\nu \sin n_i \nu l_i(1 + \sigma \xi_i) & \cos n_i \nu l_i (1 + \sigma \xi_i)
 \end{array}
 \right] = \M_i (\I + \sigma \V_i + O(\sigma^2)),
 \end{equation}
 where $\M_i = \left. \Mt_1\right |_{\sigma =0}$ is the unperturbed transfer matrix of $i$-th subinterval, and $\V_i$ has the form
 \begin{equation}
 \V_i = l_i \xi_i\left[
 \begin{array}{cc}
  0  &   1 \\[2mm]
  -n_i^2\nu^2   & 0
 \end{array}
 \right], \quad \Tr \V_i =0, \quad i=1,2,\ldots,m.
 \label{V_i}
 \end{equation}
The transfer matrix for the period (consider for simplicity the case $m=3$) is
\begin{align}
\Mt &= \Mt_3 \Mt_2 \Mt_1 = \M_3 (\I + \sigma \V_3) \M_2 (\I + \sigma \V_2)\M_1 (\I + \sigma \V_1) + O(\sigma^2) \nonu \\[2mm]
&= \M_3 \M_2 \M_1 \left[\I + \sigma \left(\V_1 + \M_1^{-1} \V_2 \M_1
+ \M_1^{-1} \M_2^{-1}\V_3 \M_2 \M_1 \right) + O(\sigma^2)\right].
\label{V3}
\end{align}
The last two terms in the parentheses are similar to the matrices $\V_2$ and $\V_3$, respectively,
and therefore by \rf{V_i} $\Tr \V_1=\Tr \M_1^{-1} \V_2 \M_1 = \Tr \M_1^{-1} \M_2^{-1}\V_3 \M_2 \M_1 =0$.
Hence,
\begin{equation}
 \tilde{\M} = {\M}\left( \I + \sigma \,\V + O(\sigma^2)\right),
 \label{Mt1a}
\end{equation}
where $\Tr \V=0$:
\begin{equation}\label{}
\V = \left[
\begin{array}{rr}
\xi & \e\\[2mm]
\zeta & -\xi
\end{array}
\right].
\end{equation}
A similar result is valid for arbitrary number $m$ of subintervals on a single period.
The entries $\xi , \e,\zeta $ are a linear combination of random variables $\xi_i,i=1,2...,m.$

Let us show that the covariance matrix $\B$ of the vector ${\bm \rho}=[\xi , \e,\zeta] $ is non-degenerate. Each matrix $\V_i$ contains only one random variable $\xi_i$, and ${\bm \rho}$ has the form 
${\bm \rho}=\sum_{i=1}^m {\bm a}_i\xi_i$, where ${\bm a}_i$ are three-dimensional vectors  defined by nonrandom
matrices. 

% Without loss of generality one can assume that $m\geq 3$ (the period can be doubled if $m=2$.) After that one can prove that there are at least three linearly independent vectors
% among $\{\bm a_i\}$ . The proof uses the fact that two adjacent subintervals have different
% refraction indices $n_i$. 

 Without loss of generality we may assume that there are at least three linearly independent vectors
among $\{{\bm a}_i\}$. If $m\geq 3$ one needs only to note that two adjacent subintervals have different
refraction indices $n_i \neq n_{i+1}$. Also, we double the period if $m=2$. Now
\begin{equation}
 \B=\E\left(\sum_{i=1}^m {\bm a}_i\xi_i\otimes\sum_{j=1}^m {\bm a}_j\xi_j\right)=\E\left( \xi_i^2 \right) \sum_{i=1}^m {\bm a}_i\otimes {\bm a}_i ,\quad \det B\neq 0.
\end{equation}

We cannot expect that the law of ${\bm \rho}$ is close to the Gaussian. However, if we choose $N\ell$ as a new period with $N \gg 1$
and assume that $N\sigma^2 \ll 1$ then the distribution of $\sum_{i=1}^{Nm} {\bm a}_i\xi_i$ will be close to the Gaussian.
Thus, an analog of Lemma \ref{lem1} holds for the case of layer thickness disorder.

\section{Comparison theorems for the Lyapunov exponents}
\label{IV}

The goal of this section is to compare the Lyapunov exponents $\gamma$ for the optical
waveguide \rf{eq2} and for the Schr\"{o}dinger equation with the white noise potential. The latter equation
can be obtained from \rf{eq2} with $A=-1$, $B=0$ when $\Delta \to 0$, and it has the form
\begin{equation}
 -\psi^{\prime \prime}(x) +  \sigma \dot{w}(x) \psi(x) = \omega^2 \psi(x).
 \label{sch}
\end{equation}
Here we use different notation for the spectral parameter in order to compare these two problems with $\omega = \omega(\nu)$.
\begin{proposition}
\label{prop2}
Lyapunov exponent $\gamma_1=\gamma_1(\sigma,\nu)$ for optical equation \rf{eq2} with $\nu$ in a neighborhood of
a band edge $\nu_0$ and Lyapunov exponent $\gamma_2=\gamma_2 (\sigma,\omega)$
for  the Schr\"{o}dinger equation \rf{sch} with $\omega=\omega(\nu)$ defined in \rf{cos} and \rf{p1} estimate each other uniformly
with respect to $\sigma$ and $\nu$, i.e.
there exist constants $c_1,c_2$ such
that $c_1 \gamma_2 \leq \gamma_1 \leq c_2 \gamma_2$ for all $0< \sigma < \sigma_0, \; | \nu-\nu_0| <\mu _0$ for some $\sigma_0,\mu_0$.
\end{proposition}
The transfer matrices $\Mt_k$ for the optical problem have the form \rf{Mt1}, where matrix $\M$ is similar to $\widehat \M$ in \rf{p1}:
\begin{equation}
 \D^{-1} \Mt_k (\nu) \D = \left[
 \begin{array}{cc}
 \cos \omega & \frac{\sin \omega}{\omega}\\
 -\omega \sin \omega & \cos \omega
 \end{array}
 \right] \left( \I + \sigma \D^{-1} \V_k \D + O(\sigma^2)
 \right), \quad \omega = \omega(\nu),
 \label{Mo}
\end{equation}
with $\V_k$ satisfying \rf{vk},\rf{B}.

The transfer matrices  $\Mt_k (\nu)$ for the Schr\"{o}dinger equation \rf{sch} also satisfy \rf{Mt1}-\rf{B} with $\M = \widehat{\M}$,
since $A=-1$ and $B=0$:
\begin{equation}
 \Mt_k (\nu) = \left[
 \begin{array}{cc}
 \cos \omega & \frac{\sin \omega}{\omega}\\
 -\omega \sin \omega & \cos \omega
 \end{array}
 \right] \left( \I + \sigma  \V_k  + O(\sigma^2)
 \right).
 \label{}
\end{equation}
Here $\V_k$ are different from those in \rf{Mo}, but both satisfy \rf{vk},\rf{B}. Now, to justify the proposition
one needs only to note that the similarity transformation of all matrices $\Mt_k$ (with the same matrix $\D$)
does not change $\gamma$ (see section \ref{II}) and apply the following
\begin{lemma}
\label{lem2}
Consider two sequences of independent identically distributed random Gaussian matrices
\begin{align}
 {\Mt}_{k}^{(i)} &= \M_0 \left[ \I + \sigma \V_{k}^{(i)} + O(\sigma^2)  \right], \quad \text{where} \nonu \\[2mm]
\V_k^{(i)} & = \left[
\begin{array}{rr}
\xi_k^{(i)} & \e_k^{(i)}  \\[2mm]
\zeta_k^{(i)} & -\xi_k^{(i)}
\end{array}
\right], \quad  \left[
\begin{array}{c}
\xi_k^{(i)} \\
\e_k^{(i)}  \\
\zeta_k^{(i)}
\end{array}
\right]\sim {\cal N} (0, \B^{(i)}), \quad i=1,2.
\end{align}
Correlation matrices $\B^{(1)}$ and $\B^{(2)}$ in both sequences are strictly positively definite matrices.
%, i.e. $\lambda_{1}\I \leq \B^{(1)}, \B^{(2)} \leq \lambda_{2} \I, \; \lambda_{1,2} >0.$
Denote by $\gamma_1$ and $\gamma_2$ the Lyapunov exponents for the sequences ${\Mt}_{k}^{(1)}$ and
${\Mt}_{k}^{(2)}, \; k = 1,2, \ldots.$
Let $\det \M_0 \neq 0$,  $0 < \sigma \leq \sigma_0$, and $\lambda_1 \I \leq \B^{(1)},\B^{(2)} \leq \lambda_2 \I $ with $\lambda_{1,2} >0$.
 Then there exist constants $c_1,c_2$ depending on $\M_0, \lambda_1, \lambda_2,$ and $\sigma_0$ such
that $c_1 \gamma_2 \leq \gamma_1 \leq c_2 \gamma_2$ uniformly with respect to
$\M_0, \B^{(1,2)},$ and $\sigma$.
\end{lemma}
The lemma has obvious physical meaning: the stronger the disorder in the waveguide, the greater is its
Lyapunov exponent. However, mathematical proof of the lemma is quite tedious and will be published elsewhere.

\section{The white noise model}
\label{V}

In the previous section we reduced calculation of the Lyapunov exponent to the product of random matrices
\begin{equation}
 \Mt (\omega) = \M_0 \left[ \I + \sigma \V_k + O(\sigma^2) \right],
\end{equation}
where
\begin{equation}
 \M_0 = \left[
 \begin{array}{cc}
 \cos \omega\ell & \frac{\sin \omega\ell}{\omega} \\[2mm]
 -\omega \sin \omega\ell & \cos \omega\ell
 \end{array}
 \right],
\end{equation}
and $\V_k, k=1,2,\ldots,$ are independent identically distributed Gaussian matrices. The exact calculation of the Lyapunov exponent for the product of random matrices is, unfortunately, quite a challenging problem. However, it can be calculated for the Schr\"{o}dinger equation with the white noise potential using the Ito calculus. To this end, we consider the model equation
\begin{equation}
-\psi^{\prime \prime}(x) + \sigma \dot{w}(x) \psi = \omega^2 \psi(x),
\quad t \in \R,
%\label{H}
\end{equation}
where $\dot{w}(x)$ is the standard white noise. Since $\dot{w}(x)$ is a generalized function, we have to define solution of the last equation. We approximate $\dot{w}(x)$
by a piecewise-constant process
\begin{equation}
 \xi_\Delta (x) = \frac{\sigma \xi_n}{\sqrt{\Delta}}, \quad x\in [n\Delta, (n+1)\Delta],
 \quad n=0,\pm 1, \pm 2, \ldots,
\end{equation}
where $\xi_n$ are independent normal random variables with $\E \,\xi =0,$ ${\sf Var}\, \xi =1$. The transfer matrix over the interval $n\Delta$ has the form
\begin{align}
 \M(n\Delta) =& \prod_{k=0}^{n-1}\left[
 \begin{array}{cc}
  \cos \left( \sqrt{\omega^2 - \frac{\sigma \xi_k}{\sqrt{\Delta}}}\,\Delta\right) & \dst \frac{\omega\sin \left( \sqrt{\omega^2 -\frac{\sigma \xi_k}{\sqrt{\Delta}}}\, \Delta\right)}{\sqrt{\omega^2 -\frac{\sigma \xi_k}{\sqrt{\Delta}}}} \\[2mm]
  \dst -\frac{\sqrt{\omega^2 -\frac{\sigma \xi_k}{\sqrt{\Delta}}}\sin \left( \sqrt{\omega^2 -\frac{\sigma \xi_k}{\sqrt{\Delta}}}\, \Delta\right)}{\omega} &  \cos \left( \sqrt{\omega^2 - \frac{\sigma \xi_k}{\sqrt{\Delta}}}\,\Delta\right)
 \end{array}
 \right] \nonu \\[2mm]
=& \prod_{k=0}^{n-1}\left[
 \begin{array}{cc}
  1 & \omega \Delta \\[2mm]
  \dst -\omega \Delta + \frac{\sigma \xi_k \sqrt{\Delta}}{\omega} & 1
 \end{array}
 \right] + O\left( \Delta^{3/2} \right).
\label{trans_asy}
\end{align}
If we introduce the polar coordinates
\begin{equation}
\begin{array}{l}
\psi (k\Delta) = r (k\Delta) \sin \theta (k\Delta), \\[2mm]
\psi^\prime (k\Delta) = \omega r (k\Delta) \cos \theta (k\Delta),
\end{array}
\label{pol}
\end{equation}
then the values of $\psi \left((k+1)\Delta\right)$ and $\omega^{-1} \psi^\prime ((k+1)\Delta)$
are related through the transfer matrix \rf{trans_asy} as follows
\begin{align}
 r \left((k+1)\Delta\right) & \left[
\begin{array}{l}
\sin \theta ((k+1)\Delta) \\[2mm]
\cos \theta ((k+1)\Delta)
\end{array}
\right] =
\left[
 \begin{array}{cc}
  1 & \omega \Delta \\[2mm]
  \dst -\omega \Delta + \frac{\sigma \xi_k \sqrt{\Delta}}{\omega} & 1
 \end{array}
 \right] r \left(k\Delta\right)
\left[
\begin{array}{l}
\sin \theta (k\Delta) \\[2mm]
\cos \theta (k\Delta)
\end{array}
\right] \nonu \\[2mm]
&= r \left(k\Delta\right) \left[
\begin{array}{c}
\sin \theta (k\Delta) + \omega \Delta \cos \theta (k\Delta) \\[2mm]
\cos \theta (k\Delta) + \dst \left(-\omega \Delta + \frac{\sigma \xi_k \sqrt{\Delta}}{\omega} \right) \sin \theta (k\Delta)
\end{array}
\right].
\label{rel1}
\end{align}
Denoting
\begin{equation}
 z(k\Delta) = \cot \theta(k\Delta)
\end{equation}
we derive from \rf{rel1}
\begin{align}
 z((k+1)\Delta) &= \frac{\dst z(k\Delta) -\omega \Delta + \frac{\sigma \xi_k \sqrt{\Delta}}{\omega}}{1+\omega \Delta z(k\Delta)} \nonu \\[2mm]
&= \left[ \dst z(k\Delta) -\omega \Delta + \frac{\sigma \xi_k \sqrt{\Delta}}{\omega}\right]
\left[ 1- \omega \Delta z(k\Delta) + O(\Delta^2) \right] \nonu \\[2mm]
&= z(k\Delta) -\omega \Delta + \frac{\sigma \xi_k \sqrt{\Delta}}{\omega} - \omega \Delta z^2 (k\Delta) + O(\Delta^{3/2}).
\label{z_delta}
\end{align}
Taking in \rf{z_delta} the limit $\Delta \to 0$ we obtain a stochastic differential equation for the limiting phase $\theta(x)$ corresponding to the white noise potential
\begin{equation}
 \frac{dz(x)}{dx}=\frac{\sigma}{\omega}\, \dot{w}(x) - \omega \,(1+z^2(x)).
 \label{phase}
\end{equation}
Similarly, for the amplitude $r(x)$ we obtain from \rf{rel1}
\begin{align}
 &r^2((k+1)\Delta) = r^2 (k\Delta)\Biggl[ \left(\sin \theta (k\Delta) + \omega \Delta \cos \theta (k\Delta)\right)^2  \nonu \\[2mm]
&+  \left( \cos \theta (k\Delta) + \Bigl(-\omega \Delta + \frac{\sigma \xi_k \sqrt{\Delta}}{\omega} \Bigr) \sin \theta (k\Delta)\right)^2 \Biggr] \nonu \\[2mm]
&= r^2 (k\Delta)\left[1 + \frac{2\sigma \xi_k \sqrt{\Delta}}{\omega} \sin \theta(k\Delta) \cos \theta(k\Delta) +  \frac{\sigma^2 \xi^2_k \Delta}{\omega^2} \sin^2 \theta(k\Delta) + O(\Delta^{3/2})\right].
\label{rel2}
\end{align}
Taking the logarithm of both sides of \rf{rel2} and sending $\Delta \to 0$ we obtain a stochastic differential equation for the amplitude $r(x)$
\begin{equation}
 \frac{d\ln r(x)}{dx}=\frac{1}{2\omega} \sin 2\theta(x)\, \dot{w}(x) -\frac{\sigma^2}{2\omega^2}
 \sin^2 \theta(x) \cos 2\theta(x).
\end{equation}
Equation \rf{phase} defines generator $\cal L$ of the diffusion process $z(x)$
\begin{equation}
 {\cal L} = \frac{\sigma^2}{2\omega^2}\,\frac{d^2}{dz^2} -\omega(1+z^2) \frac{d}{dz}.
\end{equation}
The distribution density $P(t,z_1,z_2)$ of the diffusion process satisfies the equation
\begin{equation}
 \frac{\partial P}{\partial t}= \L P \quad \text{with} \quad \left. P \right|_{t=0} = \delta(z_1 - z_2).
\end{equation}
The limiting distribution density $p$ is then given by
\begin{equation}
 \L^\ast \,p = 0
\end{equation}
or
\begin{equation}
 \frac{\sigma^2}{2\omega^3}\,\frac{d^2 p}{d z^2} + \frac{d}{d z}\left((1+z^2)\,p \right)=0.
\end{equation}
Solution of this equation has the form
\begin{equation}
 p(z) = C e^{-\Phi(z)}\int_{-\infty}^z e^{\Phi(t)} \,dt, \quad \Phi(z) = \lambda \left(z + \frac{z^3}{3}\right), \quad \lambda = \frac{2\omega^3}{\sigma^2},
\label{p}
\end{equation}
where $C$ is a normalizing constant
\begin{equation}
 \frac{1}{C} = \int_{-\infty}^\infty  e^{-\Phi(z)} \,dz\int_{-\infty}^z e^{\Phi(t)}\,dt  = \sqrt{\frac{2\pi}{\lambda}}\int_0^\infty \frac{e^{-2\lambda \left(x+\frac{x^3}{3} \right)}}{\sqrt{x}}\,dx.
\end{equation}
Now we can calculate the Lyapunov exponent
\begin{align}
 \gamma &= \lim_{L \to \infty} \frac{\E \ln(r(L)}{L}= -\frac{\sigma^2}{2\omega^2}
 \int_0^\pi p(\theta) \sin^2 \theta \cos 2\theta \,d\theta \nonu \\[2mm]
 & = \frac{\sigma^2}{2\omega^2}\, C \int_{-\infty }^{\infty }\frac{1-x^{2}}{(1+x^{2})^{2}}
e^{-\Phi (x)}dx\int_{-\infty }^{x}e^{\Phi (y)}dy.  \label{gamma1}
\end{align}

Rigorous mathematical analysis of dependence of $\gamma$ on $\omega$ and $\sigma$ is quite cumbersome and will be published elsewhere. Below we give
the summary of the result.

\subsection{Frequency near the band edge}

Let $\omega_0 \neq 0$ be a band edge that separates a gap and a non-degenerate band (i.e. point A in Fig. \ref{bands}), and let frequency $\omega$ be in the band and close to the band edge $\omega_0$ such that $|\omega - \omega_0| = \omega_\Delta \ll 1$. Then for small $\sigma$ and $\omega_\Delta$
\begin{equation}
 \gamma = C_1 \sigma^{\frac{2}{3}} \;\; \text{provided} \;\; \omega_\Delta^{\frac{3}{2}} \ll \sigma^2,
\label{nu-edge}
\end{equation}
where $C_1$ is a constant. We illustrate this dependence numerically when a small disorder is introduced in the lengths $l_1$ and $l_2$ of the waveguide in Section 3. For propagation frequency $\omega = 5.6288$ corresponding to the left edge of the band (point A in Fig. \ref{bands})
we calculated the transmission coefficient after $1000$ periods and then using \rf{t-gamma} we obtained the Lyapunov exponent $\gamma$ averaging the last $500$ periods. Results of the numerical simulation are shown in Fig. \ref{slopes} by circles. The approximating line has equation
\begin{equation}
 \lg \gamma = \frac{2}{3} \lg \sigma + 0.15
\end{equation}
which is in excellent agreement with the asymptotics \rf{nu-edge}.
The same dependence was confirmed numerically in \cite{MO07}.

\begin{figure}[ht]
\begin{center}
\includegraphics[width=0.8\textwidth, angle=0]{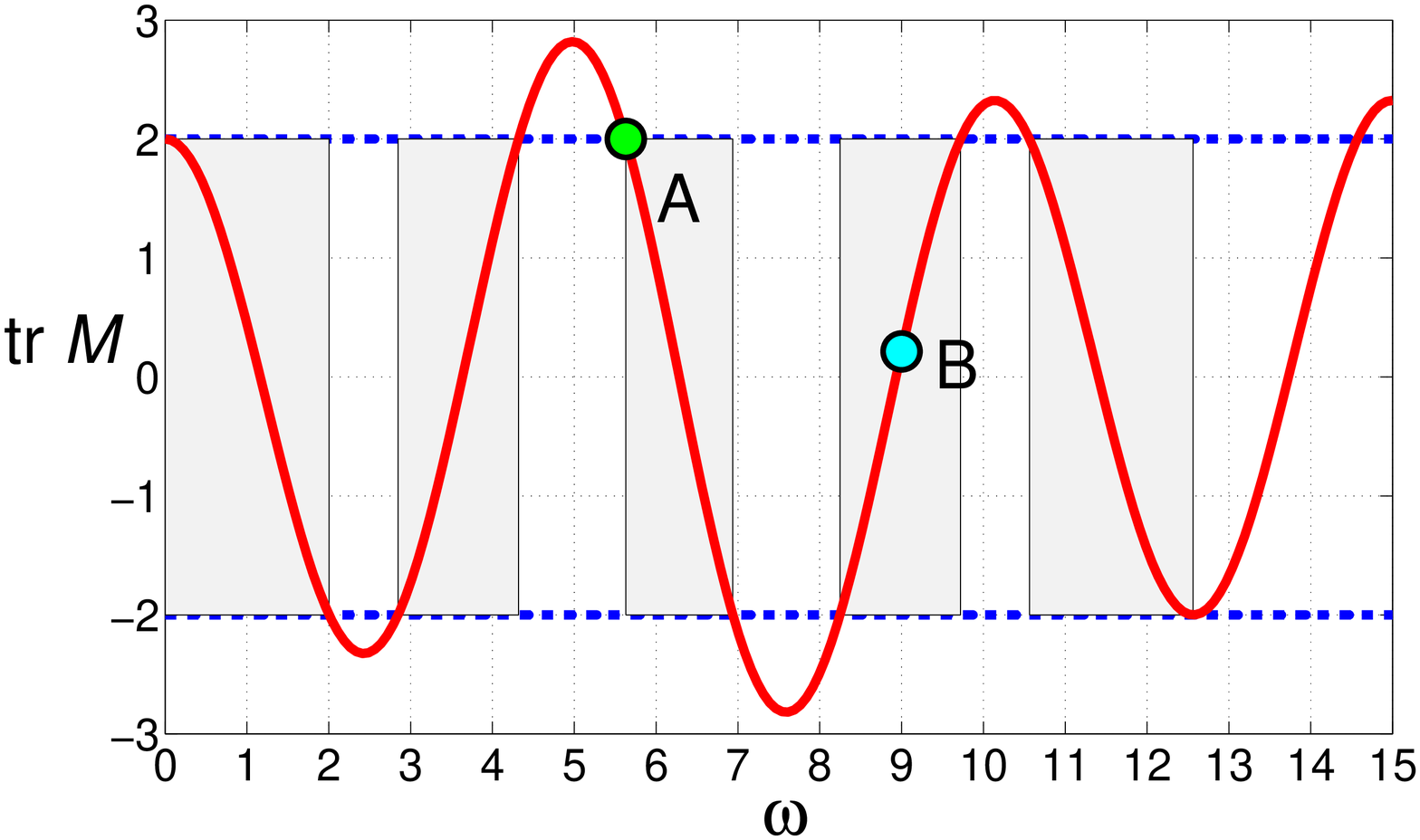}
\end{center}
\caption{The discriminant curve and bandgap structure of model 1 with $n_1=1,\;n_2=2.5,\;l_1=1,\;l_2=0.1$. Point A corresponds to the left band edge $\omega = 5.6288$ while point B with $\omega = 9$ lies in the middle of the band.}
\label{bands}
\end{figure}

\begin{figure}[ht]
\begin{center}
\includegraphics[width=0.8\textwidth, angle=0]{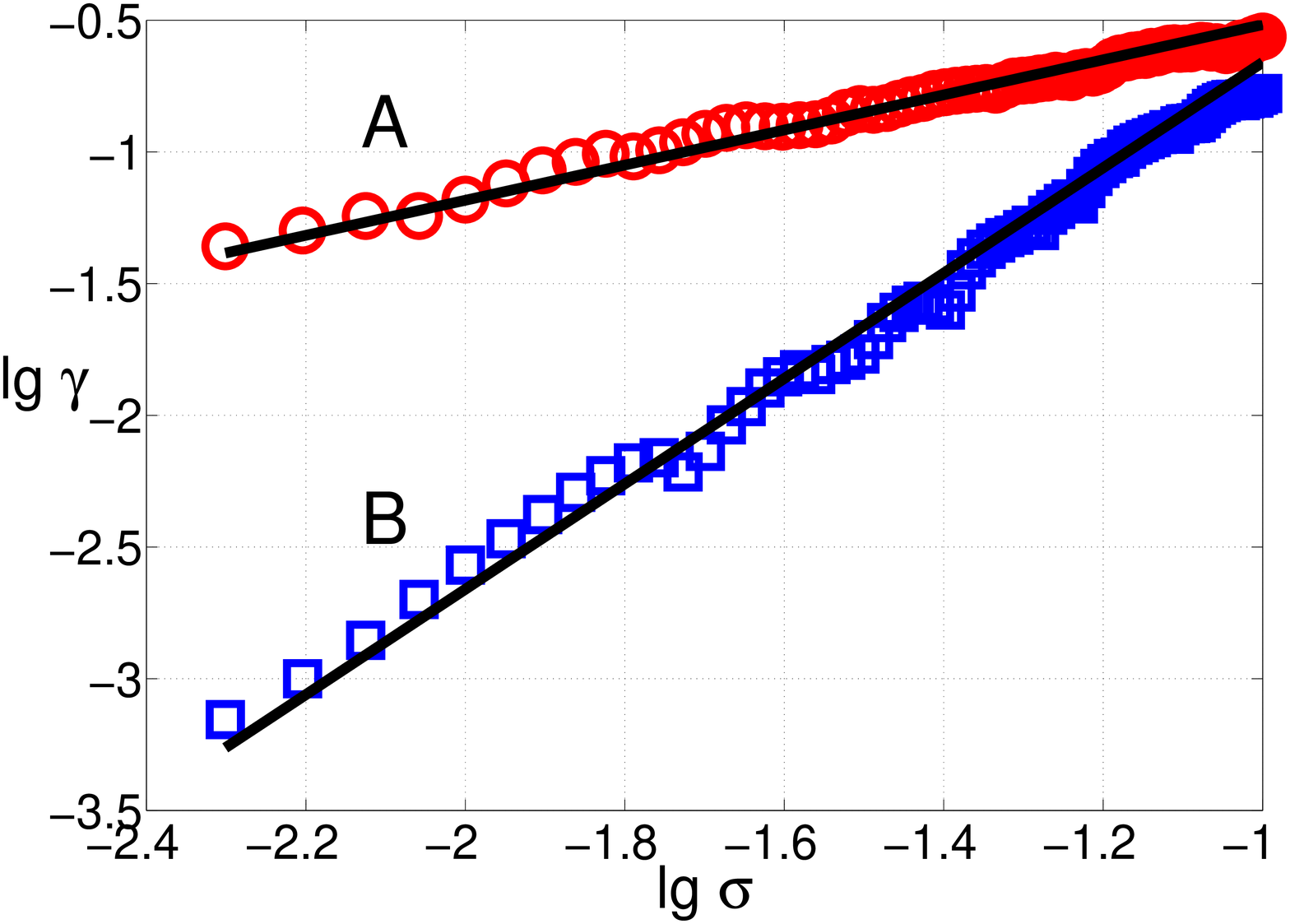}
\end{center}
\caption{Numerical simulation of dependence of the Lyapunov exponent $\gamma $ on the strength of disorder $\sigma$. Circles and squares correspond to perturbation of points A and B in Fig. \ref{bands}, respectively, and approximated by the lines with the slopes $2/3$ and $2$.}
\label{slopes}
\end{figure}

\subsection{Frequency in the bulk of the band}

If frequency $\omega$ is located in the band far from the band edges then the character of asymptotics of $\gamma$ is completely different. In this case
\begin{equation}
\gamma = C_2 \sigma^2 \;\; \text{provided} \;\; \sigma \ll 1.
\label{nu-bulk}
\end{equation}
Here $C_2$ is a constant independent of $\sigma$. We illustrate this dependence numerically for frequency $\omega=9$ (point B in Fig. \ref{bands}). The result is shown in Fig. \ref{slopes} by squares. Dependence of $\lg \gamma$ on $\lg \sigma$ is then approximated by the line
\begin{equation}
 \lg \gamma = 2 \lg \sigma + 1.34
\end{equation}
which is well consistent with the asymptotics \rf{nu-bulk}.

\section{Conclusions}
\label{VI}
We have studied the influence of disorder on transmission through periodic
waveguides described by the Schr\"{o}dinger and optical equations. While the bands of the Schr\"{o}dinger operator have regular structure, it is shown that the bands of the optical operator do not possess this property, and the lengths of its gaps exhibits a chaotic behavior as the frequency increases.
%We have demonstrated that both equations can be represented as a canonical system with a symplectic transfer matrix.
Assuming that the transfer matrix of perturbed waveguide has the form $\tilde{\M}_k = {\M}\left( \I + \sigma \,\V_k + O(\sigma^2)\right)$ we have shown
that this representation is valid for two particular models of disordered waveguides. 
Using the result that Lyapunov exponents of systems with transfer matrices in
the canonical form estimate each other, 
%we have proven that Lyapunov 
%exponents of systems with such form of transfer matrix estimate each other.
we have estimated the Lyapunov exponent of the optical waveguide through that of the Schr\"{o}dinger equation with the white noise potential.
Exact dependence of the Lyapunov exponent $\gamma$ was established as a function of frequency $\omega$ and intensity of the disorder $\sigma$.
When the frequency $\omega$ lies in the bulk of the band then $\gamma \sim \sigma^2$, while near the band edge $\gamma$ has the order $\gamma \sim \sigma^{2/3}$. Thus, small disorder drastically reduces transmission of the waveguide if the frequencies are located near non-degenerate band edges.

\end{document}